\documentclass[11pt,english]{article}
\usepackage{lmodern}

\usepackage[T1]{fontenc}
\usepackage[latin9]{inputenc}
\usepackage{geometry}
\usepackage{xcolor}
\usepackage{verbatim}
\usepackage{float}
\usepackage{amsmath}
\usepackage{amsthm}
\usepackage{amssymb}
\usepackage{graphicx}
\usepackage{setspace}
\usepackage[authoryear]{natbib}
\makeatletter
\theoremstyle{definition}
\newtheorem{defn}{\protect\definitionname}
\theoremstyle{plain}
\newtheorem{thm}{\protect\theoremname}
\theoremstyle{plain}
\newtheorem{lem}{\protect\lemmaname}
\theoremstyle{plain}
\newtheorem{prop}{\protect\propositionname}
\theoremstyle{plain}
\newtheorem{cor}{\protect\corollaryname}

\setcitestyle{round}
\usepackage{lscape}
\usepackage{longtable}

\usepackage{graphicx}
\usepackage{tabularx}

\usepackage{array,booktabs,ragged2e}
\newcolumntype{C}[1]{>{\centering\arraybackslash}p{#1}}
\newcolumntype{J}[1]{>{\justify\arraybackslash}p{#1}}
\newcolumntype{R}[1]{>{\RaggedLeft\arraybackslash}p{#1}}
\newcolumntype{Q}[1]{>{\columncolor{Gray}\RaggedLeft\arraybackslash}p{#1}}
\newcolumntype{L}[1]{>{\RaggedRight\arraybackslash}p{#1}}
\newcolumntype{G}{@{\extracolsep{0.5cm}}l@{\extracolsep{0pt}}}%
\usepackage{multirow}

\let\OLDthebibliography\thebibliography
\renewcommand\thebibliography[1]{
  \OLDthebibliography{#1}
  \setlength{\parskip}{0pt}
  \setlength{\itemsep}{0pt plus 0.3ex}
}

\AtBeginDocument{

}

\makeatother

\usepackage{babel}
\providecommand{\corollaryname}{Corollary}
\providecommand{\definitionname}{Definition}
\providecommand{\lemmaname}{Lemma}
\providecommand{\propositionname}{Proposition}
\providecommand{\theoremname}{Theorem}

\begin{document}
\title{A New Parametrization of Correlation Matrices\thanks{We are grateful for many valuable comments made by Immanuel Bomze,
Bo Honoré, Ulrich Müller, Georg Pflug, Werner Ploberger, Rogier Quaedvlieg,
and Christopher Sims, as well as many conference and seminar participants. }\emph{\normalsize{}\medskip{}
}}
\author{\textbf{Ilya Archakov}$^{a}$\textbf{ and Peter Reinhard Hansen}$^{b}$\textbf{}\thanks{Address: University of North Carolina, Department of Economics, 107
Gardner Hall Chapel Hill, NC 27599-3305}\bigskip{}
\\
{\normalsize{}$^{a}$}\emph{\normalsize{}University of Vienna\smallskip{}
}\\
{\normalsize{}$^{b}$}\emph{\normalsize{}University of North Carolina
\& Copenhagen Business School\medskip{}
}}
\date{\emph{\normalsize{}\today}}
\maketitle
\begin{abstract}
We introduce a novel parametrization of the correlation matrix. The
reparametrization facilitates modeling of correlation and covariance
matrices by an unrestricted vector, where positive definiteness is
an innate property. This parametrization can be viewed as a generalization
of Fisher's $Z$-transformation to higher dimensions and has a wide
range of potential applications. An algorithm for reconstructing the
unique $n\times n$ correlation matrix from any vector in $\mathbb{R}^{n(n-1)/2}$
is provided, and we derive its numerical complexity. 

\bigskip{}
\end{abstract}
\textit{\small{}Keywords:}{\small{} Correlation Matrix, Covariance
Modeling, Fisher Transformation.}{\small\par}

\noindent \textit{\small{}JEL Classification:}{\small{} C10; C22;
C58 \newpage}{\small\par}

\section{Introduction}

We propose a new way to parametrize a covariance matrix that ensures
positive definiteness without imposing additional restrictions. The
central element of the parametrization is the matrix logarithmic transformation
of the correlations matrix, $\log C$, whose lower off-diagonal elements
are stacked into the vector $\gamma=\gamma(C)$. We show that this
transformation defines a one-to-one correspondence between the set
of $n\times n$ non-singular correlation matrices and $\mathbb{R}^{n(n-1)/2}$,
and we propose a fast algorithm for the computation of the inverse
mapping.\footnote{Code for this algorithm (Julia, Matlab, Ox, Python, and R) is provided
in the Web Appendix.} In the bivariate case, $n=2$, $\gamma(C)$ is identical to the Fisher
transformation, and simulation results suggest that $\gamma(C)$ inherits
some of the attractive properties of the Fisher transformation when
$n>2$. 

Our results show that a non-singular $n\times n$ covariance matrix
can be expressed as a unique vector in $\mathbb{R}^{n(n+1)/2}$ that
consists of the $n$ log-variances and $\gamma$. This facilitates
the modeling of covariance matrices in terms of an unrestricted vector
in $\mathbb{R}^{n(n+1)/2}$. In models with dynamic covariance matrices,
such as multivariate GARCH models and stochastic volatility models,
the parametrization offers a new way to structure multivariate volatility
models. The vector representation offers new ways to regularizing
large covariance matrices by imposing structure on $\gamma$. The
new parametrization can also be used to specify distributions on the
space of non-singular correlation matrices and covariance matrices.
This could be useful in multivariate stochastic volatility models
and Bayesian analysis.

It is convenient to reparametrize a covariance matrix as a vector
that is unrestricted in $\mathbb{R}^{d}$, and the literature has
proposed several methods to this end, see \citet{PinheiroBates:1996}.
These methods include the Cholesky decomposition, the spherical trigonometric
transformation, transformations based on partial correlation vines,
and methods based on the spectral representation, such as the matrix
logarithm, see e.g. \citet{KurowickaCooke:2003}. The matrix logarithm
has been used in the modeling of covariance matrices in \citet{LeonardHsu:1992}
and \citet{ChiuLeonardTsui:1996}. In GARCH and stochastic volatility
models it was used in \citet{Kawakatsu:2006}, \citet{IshiharaOmoriAsai:2016},
and \citet{AsaiSo:2015}, and \citet{BauerVorkink:2011} used the
matrix logarithm for modeling and forecasting of realized covariance
matrices. The transformation also emerges as a special case of the
Box-Cox transformation, see \citet{Weigand:2014} for an application
to realized covariance matrices. 

We do not apply the matrix logarithm to covariance matrices, but to
correlation matrices. Modeling the correlation matrix separately from
the individual variances is commonly done in multivariate GARCH models,
see e.g. \citet{Bollerslev:1990}, \citet{Engle2002}, \citet{TseTsui:2002},
and \citet{EngleKelly2012}. The new parametrization can be used to
define a new family of multivariate GARCH models, that need not impose
additional restrictions beyond positivity. Additional structure can
be imposed, if so desired, and we provide examples of this in Section
\ref{sec:Auxiliary-Results}. The new parametrization can also be
used in dynamic models of multivariate volatility that make use of
realized measures of volatility. Such as those in \citet{Liu_2009},
\citet{ChiriacVoev:2011}, \citet{Golosnoy_Gribisch_Liesenfeld_2012},
\citet{BauwensStortiViolante:2012}, \citet{NoureldinShephardSheppard:2012},
\citet{HansenLundeVoev:2014}, and \citet{GorgiHansenJanusKoopman:2019}.

The paper is organized as follows. We introduce and motivate the new
parametrization of correlation matrices in Section 2 by relating it
to the Fisher transformation. We present the main theoretical results
in Section \ref{sec:Theoretical-Framework}, auxiliary results in
Section \ref{sec:Auxiliary-Results}, and analyze the algorithm for
evaluating the inverse mapping, $C(\gamma)$, in Section \ref{sec:Algorithm}.
We conclude and summarize in Section \ref{sec:Concluding-Remarks}.
All proofs are given in the Appendix, and additional results and computer
code are collected in the Web Appendix, see \citet{ArchakovHansen:CorrAppendix}.

\section{Motivation}

We motivate the proposed method by considering a non-singular $2\times2$
covariance matrix, with variances $\sigma_{1}^{2}$ and $\sigma_{2}^{2}$
and the correlation $\rho=\sigma_{12}/(\sigma_{1}\sigma_{2})\in(-1,1)$.
This matrix can be reparametrized as the vector $v=(\log\sigma_{1},\log\sigma_{2},\digamma(\rho))^{\prime}$,
where $\digamma(\rho)=\tfrac{1}{2}\log\tfrac{1+\rho}{1-\rho}$ is
the Fisher transformation. Because any $v\in\mathbb{R}^{3}$ maps
to a unique non-singular covariance matrix this defines a one-to-one
mapping between the non-singular covariance matrices and $\mathbb{R}^{3}$.
The vector parametrization is convenient because a positive definite
covariance matrix is guaranteed without imposing additional restrictions. 

We seek a similar parametrization of covariance matrices when $n>2$.
Specifically, a mapping so that 1) Any non-singular covariance matrix,
$\Sigma$, maps to a unique vector $v=\nu(\Sigma)\in\mathbb{R}^{d}$;
2) Any vector $v\in\mathbb{R}^{d}$ maps to a unique covariance matrix
$\Sigma=\nu^{-1}(v)$; 3) The parametrization, $v=\nu(\Sigma)$, is
``invariant'' to the ordering of the variables that define $\Sigma$;
and 4) the elements of $v$ are easily interpretable.

The parametrization, $v=(\log\sigma_{1},\log\sigma_{2},\tfrac{1}{2}\log\frac{1+\rho}{1-\rho})^{\prime}$,
has all these above properties. The Cholesky representation is not
invariant to the ordering of variables. The matrix logarithm transformation
of covariance matrix, $\log\Sigma$, satisfies the first three three
properties, but the resulting elements are difficult to interpret,
because they depend non-linearly on all elements of $\Sigma$. For
$n>2$ one could consider the element-wise Fisher transformations
of every correlation, but this will not satisfy the second property.\footnote{For instance, the inverse Fisher transformation of, $-2$, $0$, and
$\tfrac{1}{2}$ will result in three correlations that, combined,
will produce a ``correlation matrix'' with a negative eigenvalue. } 

Returning to the case with a $2\times2$ correlation matrix. We observe
that the Fisher transformation appears as the off-diagonal elements
when we take the matrix-logarithm of an $2\times2$ correlation matrix:\setlength{\belowdisplayskip}{5pt} \setlength{\belowdisplayshortskip}{5pt} \setlength{\abovedisplayskip}{-8pt} \setlength{\abovedisplayshortskip}{-8pt}

\[
\log\left(\begin{array}{cc}
1 & \rho\\
\rho & 1
\end{array}\right)=\left(\begin{array}{cc}
\frac{1}{2}\log(1-\rho^{2}) & \frac{1}{2}\log\tfrac{1+\rho}{1-\rho}\\
\frac{1}{2}\log\tfrac{1+\rho}{1-\rho} & \frac{1}{2}\log(1-\rho^{2})
\end{array}\right).
\]
In this paper, we propose to parametrize correlation matrices using
the off-diagonal elements of $\log C$, so that an $n\times n$ covariance
matrix, $\Sigma$, is parametrized by the $n$ log-variances and the
$n(n-1)/2$ off-diagonal elements of $\log C$, denoted by $\gamma$.
We will show that this parametrization satisfies the first three objectives
stated above. The fourth objective is partly satisfied, because $n$
elements of $v$ will correspond to the $n$ individual variances,
whereas the remaining elements parametrize the underlying correlation
matrix. The Fisher transformation has attractive finite sample properties
(variance stabilizing and skewness reducing) and $\gamma$ is identical
to the Fisher transformation when $n=2$. Simulation results in the
Web Appendix suggest that the off-diagonal elements of $\log C$ inherit
some of these properties when $n>2$. 

\section{Theoretical Framework and Main Results\label{sec:Theoretical-Framework}}

We need to introduce some useful notation and terminology. The operator,
$\mathrm{diag}(\cdot)$, is used in two ways. When the argument is
a vector, $v=(v_{1},\ldots,v_{n})^{\prime}$, then $\mathrm{diag}(v)$
denotes the $n\times n$ diagonal matrix with $v_{1},\ldots,v_{n}$
along the diagonal, and when the argument is a square matrix, $A\in\mathbb{R}^{n\times n}$,
then $\mathrm{diag}(A)$ extracts the diagonal of $A$ and returns
it as a column vector, i.e. $\mathrm{diag}(A)=(a_{11},\ldots,a_{nn})^{\prime}\in\mathbb{R}^{n}$.
The matrix exponential is defined by $e^{A}=\sum_{k=0}^{\infty}\tfrac{A^{k}}{k!}$
for any matrix $A$. For any symmetric matrix, $A$, we have $e^{A}=Q\mathrm{diag}(e^{\lambda_{1}},\ldots,e^{\lambda_{n}})Q^{\prime}$,
where $A=Q\Lambda Q^{\prime}$, with $Q$ being an orthonormal matrix,
i.e. $Q^{\prime}Q=I$, and $\Lambda=\mathrm{diag}(\lambda_{1},\ldots,\lambda_{n})$
where $\lambda_{1},\ldots,\lambda_{n}$ are the eigenvalues of $A$.
The general definition of the matrix logarithm is more involved, see
\citet{Higham2008}, but for a symmetric positive definite matrix,
we have that $\log A=Q\log\Lambda Q^{\prime}$, where $\log\Lambda=\mathrm{diag}(\log\lambda_{1},\ldots,\log\lambda_{n})$. 

We use $\mathrm{vecl}(A)$ to denote the vectorization operator of
the lower off-diagonal elements of $A$. For a non-singular correlation
matrix, $C$, we let $G=\log C$ denote the logarithmically transformed
correlation matrix, and let $F$ be the matrix of element-wise Fisher
transformed correlations (whose diagonal is unspecified). The vector
of correlation coefficients is denoted by $\varrho=\mathrm{vecl}C$,
and the corresponding elements of $G$ and $F$ are denoted by $\gamma=\mathrm{vecl}G$
and $\phi=\mathrm{vecl}F$, respectively.
\begin{defn}[New Parametrization of Correlation Matrices]
For a non-singular correlation matrix, $C$, we introduce the following
parametrization: $\gamma(C):=\mathrm{vecl}(\log C)$.
\end{defn}
Because $\gamma(C)$ discards the diagonal elements of $\log C$,
it is relevant to ask: Can $C$ be reconstructed from $\gamma$ alone?
If so: Is the reconstructed correlation matrix unique for all $\gamma$?
To formalize this inversion problem, we introduce the following operator.
For an $n\times n$ matrix, $A$, and any vector $x\in\mathbb{R}^{n}$
we let $A[x]$ denote the matrix $A$ where $x$ has replaced its
diagonal. So it follows that $\mathrm{vecl}(A)=\mathrm{vecl}(A[x])$
and that $x=\mathrm{diag}(A[x])$. 

\subsection{Main Theoretical Results}
\begin{thm}
\label{thm:OneToOne}For any real symmetric matrix, $A\in\mathbb{R}^{n\times n}$,
there exists a unique vector, $x^{\ast}\in\mathbb{R}^{n}$, such that
$e^{A[x^{\ast}]}$ is a correlation matrix. 
\end{thm}
This shows that any vector in $\mathbb{R}^{n(n-1)/2}$ maps to a unique
correlation matrix, so that $\gamma(C)$ is a one-to-one correspondence
between $\mathcal{C}_{n}$ and $\mathbb{R}^{n(n-1)/2}$, where $\mathcal{C}_{n}$
denotes the set of non-singular correlation matrices.\footnote{Singular correlation matrices with known null space can be parametrized
applying the transformation to a full rank principal sub-matrix. We
do not explore this topic in this paper.} The inverse mapping, denoted $C(\gamma)$, is therefore well defined.

Next, we outline the structure of the proof of Theorem \ref{thm:OneToOne},
because it provides intuition for the algorithm that is used to reconstruct
$C$ from $\gamma$. 

Consider the mapping $g:\mathbb{R}^{n}\curvearrowright\mathbb{R}^{n}$,
$g(x)=x-\log\mathrm{diag}(e^{A[x]})$, where the logarithm is applied
element-wise to vector of diagonal elements. Because $e^{A[x]}$ is
a correlation matrix if and only if all diagonal elements are equal
to one, the requirement is simply $g(x^{\ast})=x^{\ast}$. So Theorem
\ref{thm:OneToOne} is equivalent to the statement that $g$ has a
unique fixed-point for any matrix $A$. This follows by showing the
following result and applying Banach fixed-point theorem.
\begin{lem}
\label{lem:gContraction}The mapping $g$ is a contraction for any
symmetric matrix $A$. 
\end{lem}
The proof of Lemma \ref{lem:gContraction} entails deriving the Jacobian
for $g$, denoted $\nabla g$, and showing that all its eigenvalues
are less than one in absolute value. The largest eigenvalue of $\nabla g$
is, not surprisingly, key for the algorithm that reconstructs $C$
from $\gamma$.

\subsection{Invariance to Reordering of Variables}

The mapping, $\gamma(C)$, is invariant to a reordering of variables
that define $C$, in the sense that a permutation of the variables
that define $C$ will merely result in a permutation of the elements
of $\gamma$. The formal statement is as follows.
\begin{prop}
\label{prop:Permutate}Suppose that $C_{x}=\mathrm{corr}(X)$ and
$C_{y}=\mathrm{corr}(Y)$, where the elements of $X$ is a permutation
of the elements of $Y$. Then the elements of $\gamma_{x}=\gamma(C_{x})$
is a permutation of the elements of $\gamma_{y}=\gamma(C_{y})$.
\end{prop}

\subsection{An Algorithm for Computing $C(\gamma)$}

Evidently, the solution, $x^{\ast}$, must be such that the diagonal
elements of the matrix, $e^{A[x^{\ast}]}$, are all equal to one.
Equivalently, $\log\mathrm{diag}(e^{A[x^{\ast}]})=0\in\mathbb{R}^{n}$,
where the logarithm is applied element-wise to the vector of diagonal
elements. This observation motivates the following iterative procedure
for determining $x^{\ast}$:\setlength{\belowdisplayskip}{5pt} \setlength{\belowdisplayshortskip}{5pt} \setlength{\abovedisplayskip}{-2pt} \setlength{\abovedisplayshortskip}{-2pt}
\begin{cor}
\label{cor:Algorithm}Consider the sequence,
\begin{align*}
x_{(k+1)}=x_{(k)}-\log\mathrm{diag}(e^{A[x_{(k)}]}), & \qquad k=0,1,2,\ldots
\end{align*}
with an arbitrary initial vector $x_{(0)}\in\mathbb{R}^{n}$. Then
$x_{(k)}\rightarrow x^{\ast}$, where $x^{\ast}$ is the solution
in Theorem \ref{thm:OneToOne}.
\end{cor}
In practice we find that the simple algorithm, proposed in Corollary
\ref{cor:Algorithm}, converges very fast. This is demonstrated in
Section \ref{sec:Algorithm} for matrices with dimension up to $n=100$.
The result in Theorem \ref{thm:OneToOne} and the algorithm in Corollary
\ref{cor:Algorithm} are easily adapted to a covariance matrix with
known diagonal elements, as we show in Section \ref{subsec:Results-for-Covariance-known-diag}. 

\subsection{Asymptotic Distribution of $\hat{\gamma}$\label{subsec:Asymptotic-Properties}}

Next, we derive the asymptotic distributions of $\hat{\gamma}$ and
the vector of Fisher transformed correlations, $\hat{\phi}$, by deducing
them from those of the empirical correlation matrix. 

Suppose that $\sqrt{T}(\hat{C}-C)\overset{d}{\rightarrow}N(0,\Omega)$,
as $T\rightarrow\infty$. The asymptotic covariance matrix, $\Omega=\mathrm{avar}(\text{vec}(\hat{C}))$,
will be singular because $\hat{C}$ is symmetric and has constant
diagonal elements. Convenient closed-form expressions for $\Omega$
is available in special cases, see e.g. \citet{Neudecker_Wesselman_1990},
\citet{Nel_1985}, and \citet{Browne_Shapiro_1986}.

For the vector of correlation coefficients, $\hat{\varrho}=\text{vecl}(\hat{C})$,
it follows that $\sqrt{T}(\hat{\varrho}-\varrho)\overset{d}{\rightarrow}N(0,\Omega_{\varrho})$,
as $T\rightarrow\infty$, where $\Omega_{\varrho}=E_{l}\Omega E_{l}^{\prime}$
and $E_{l}$ is an elimination matrix, characterized by $\mathrm{vecl}[M]=E_{l}\mathrm{vec}[M]$
for any $n\times n$ matrix $M$. For the element-wise Fisher transform,
the asymptotic distribution reads\setlength{\belowdisplayskip}{5pt} \setlength{\belowdisplayshortskip}{5pt} \setlength{\abovedisplayskip}{-10pt} \setlength{\abovedisplayshortskip}{-10pt}

\begin{equation}
\sqrt{T}(\hat{\phi}-\phi)\overset{d}{\rightarrow}N(0,\Omega_{\phi}),\qquad\Omega_{\phi}=D_{c}E_{l}\Omega E_{l}^{\prime}D_{c},\label{eq:avar-phi}
\end{equation}
where $D_{c}=\text{diag}\Bigl(\frac{1}{1-c_{i}^{2}},\frac{1}{1-c_{2}^{2}},\ldots,\frac{1}{1-c_{d}^{2}}\Bigl)$
and $c_{i}$ is an $i$-th element of $c=\text{vecl}(C)\in\mathbb{R}^{d}$
with $d=n(n-1)/2$, whereas the asymptotic distribution of the new
parametrization of correlation matrices, can be shown to be

\begin{equation}
\sqrt{T}(\hat{\gamma}-\gamma)\overset{d}{\rightarrow}N(0,\Omega_{\gamma}),\qquad\Omega_{\gamma}=E_{l}A^{-1}\Omega A^{-1}E_{l}^{\prime},\label{eq:avar-gamma}
\end{equation}
where $A$ is a Jacobian matrix, such that $\partial\mathrm{vec}(C)=A\,\partial\mathrm{vec}(\log C)$.
The expression for $A$ is given in the Appendix, see (\ref{eq:JacobianA-expression})-(\ref{eq:xi_defined}),
and is taken from \citet{LintonMcCrorie:1995}.

In a classical setting where $\hat{C}$ is computed from i.i.d. random
vectors, the diagonal elements of $\Omega_{\phi}$ are all equal to
one. This demonstrates the variance stabilizing property of the Fisher
transformation. The transformation $\gamma(C)$ is, evidently, not
variance stabilizing when $n>2$, except in special cases. However,
it does appear to reduce skewness, which is another attribute of the
Fisher transformation.

The two expressions for the asymptotic variances, $\Omega_{\phi}$
and $\Omega_{\gamma}$, are not easily compared unless $\Omega$ is
known. Here we will compare them in the situation where $\hat{C}$
is computed from $X_{i}\sim\text{iid}N_{3}(0,\Sigma)$, for four different
choices for $\Sigma$. Scaling the elements of $X_{i}$ does not affect
the limit distributions for $\hat{\varrho}$, $\hat{\phi}$, and $\hat{\gamma}$.
So we can, without loss of generality, focus on the case where $\Sigma=C$.
\begin{table}[H]
\begin{centering}
\scriptsize
\setlength\arraycolsep{2pt}
\begin{tabular*}{\textwidth}{c@{\extracolsep{\fill}}c@{\extracolsep{\fill}}c@{\extracolsep{\fill}}c@{\extracolsep{\fill}}c}
\\[-0.6cm]
\toprule
\\[0mm]
  $\Sigma=C$ & $\text{avar}(\hat\varrho)$ & $\text{avar}(\hat\phi)=$ & $\text{avar}(\hat\gamma)$  & $\text{acorr}(\hat\gamma)$  \\[3mm]
      &                             & $\text{acorr}(\hat\varrho)=\text{acorr}(\hat\phi)$  & \\[3mm]
\midrule
\\[0.0cm]
$\left(\begin{array}{ccc}
1&\bullet&\bullet\\
0&1&\bullet\\
0&0&1
\end{array}\right)$
&
$\left(\begin{array}{ccc}
1.000 & \bullet & \bullet\\
0 & 1.000 & \bullet\\
0 & 0 & 1.000
\end{array}\right)$
&
$\left(\begin{array}{ccc}
1.000 & \bullet & \bullet\\
0 & 1.000 & \bullet\\
0 & 0 & 1.000
\end{array}\right)$
&
$\left(\begin{array}{ccc}
1.000 & \bullet & \bullet\\
0 & 1.000 & \bullet\\
0 & 0 & 1.000
\end{array}\right)$
&
$\left(\begin{array}{ccc}
1.000 & \bullet & \bullet\\
0 & 1.000 & \bullet\\
0 & 0 & 1.000
\end{array}\right)$  \\
\\[0.0cm]

$\left(\begin{array}{ccc}
1&\bullet&\bullet\\
0.5&1&\bullet\\
0.25&0.5&1
\end{array}\right)$
 & $\left(\begin{array}{ccc}
0.562 & \bullet & \bullet\\
0.316 & 0.879 & \bullet\\
0.070 & 0.316 & 0.562
\end{array}\right)$ & $\left(\begin{array}{ccc}
1.000 & \bullet & \bullet\\
0.450 & 1.000 & \bullet\\
0.125 & 0.450 & 1.000
\end{array}\right)$ & $\left(\begin{array}{ccc}
0.966 & \bullet & \bullet\\
0.018 & 0.962 & \bullet\\
0.021 & 0.018 & 0.966
\end{array}\right)$  
& $\left(\begin{array}{ccc}
1.000 & \bullet & \bullet\\
0.018 & 1.000 & \bullet\\
0.021 & 0.018 & 1.000
\end{array}\right)$  \\
\\[0.0cm]

$\left(\begin{array}{ccc}
1&\bullet&\bullet\\
0.9&1&\bullet\\
0.81&0.9&1
\end{array}\right)$
& $\left(\begin{array}{ccc}
0.036 & \bullet & \bullet\\
0.046 & 0.118 & \bullet\\
0.015 & 0.046 & 0.036
\end{array}\right)$
&
$\left(\begin{array}{ccc}
1.000 & \bullet & \bullet\\
0.698 & 1.000 & \bullet\\
0.405 & 0.698 & 1.000
\end{array}\right)$
&
$\left(\begin{array}{ccc}
0.817 & \bullet & \bullet\\
0.081 & 0.860 & \bullet\\
0.093 & 0.081 & 0.817
\end{array}\right)$
&
$\left(\begin{array}{ccc}
1.000 & \bullet & \bullet\\
0.097 & 1.000 & \bullet\\
0.114 & 0.097 & 1.000
\end{array}\right)$  \\
\\[0.0cm]

$\left(\begin{array}{ccc}
1&\bullet&\bullet\\
0.99&1&\bullet\\
0.98&0.99&1
\end{array}\right)$
& \hspace{-1mm}$\tfrac{1}{10}\hspace{-1mm}\left(\begin{array}{ccc}
0.004 & \bullet & \bullet\\
0.006 & 0.016 & \bullet\\
0.002 & 0.006 & 0.004
\end{array}\right)$ &
$\left(\begin{array}{ccc}
1.000 & \bullet & \bullet\\
0.745 & 1.000 & \bullet\\
0.490 & 0.745 & 1.000
\end{array}\right)$ &
$\left(\begin{array}{ccc}
0.756 & \bullet & \bullet\\
0.106 & 0.793 & \bullet\\
0.134 & 0.106 & 0.756
\end{array}\right)$
&
$\left(\begin{array}{ccc}
1.000 & \bullet & \bullet\\
0.137 & 1.000 & \bullet\\
0.178 & 0.137 & 1.000
\end{array}\right)$  \\
\\
\bottomrule
\end{tabular*}
\par\end{centering}
\caption{\label{tab:AvarAcorr}Asymptotic covariance and correlation matrices
for $\hat{\varrho}$, $\hat{\phi}$ and $\hat{\gamma}$, for four
different correlation matrices. The diagonal elements of the asymptotic
variance matrix for $\hat{\phi}$ are all one, so it is also the asymptotic
correlation matrix for $\hat{\phi}.$ Because $\hat{\phi}$ is based
on an element-by-element transformation of the corresponding elements
of $\hat{\varrho}$, it is also the asymptotic correlation matrix
for $\hat{\varrho}$.}
\end{table}

The asymptotic variance and correlation matrices for the three vectors,
$\hat{\varrho}$, $\hat{\phi}$ and $\hat{\gamma}$, are reported
in Table \ref{tab:AvarAcorr}. The true correlation matrix is given
in the first column of Table \ref{tab:AvarAcorr}. The asymptotic
variance of the correlation coefficient, $\hat{\varrho}_{j}$, is
$(1-\varrho_{j}^{2})^{2}$, which defines the diagonal elements of
$\Omega_{\varrho}$, and the element-wise Fisher transformation ensures
that $\mathrm{avar}(\hat{\phi}_{j})=1$ for all $j=1,\dots,n$. However,
we observe a high degree of correlation across the elements of $\hat{\phi}$.
The asymptotic correlation matrix for $\hat{\phi}$ is, in fact, identical
to that of the empirical correlations, $\hat{\varrho}$, because the
Fisher transformation is an element-by-element transformation. Its
Jacobian, $D_{c}=\partial\phi/\partial\varrho$, is therefore a diagonal
matrix. Consequently, the asymptotic correlations are unaffected by
the element-wise Fisher transformation, and $\mathrm{acorr}(\hat{\varrho})=\mathrm{acorr}(\hat{\phi})$.
While the diagonal elements of $\Omega_{\phi}$ are invariant to $C$,
this is not the case for the diagonal elements of $\Omega_{\gamma}$,
but it is interesting to note that the asymptotic correlations between
elements of $\hat{\gamma}$ tend to be relatively small, and close
to zero when the correlations in $C$ are small.

Simulation results in the Web Appendix suggest that the elements of
$\hat{\gamma}$ tend to be weakly correlated, and that $\gamma(C)$
reduces skewness, as is the case for the Fisher transformation. Empirical
results in \citet{ArchakovHansenLundeMRG} show that the empirical
distribution of transformed realized correlation matrices is well
approximated by a Gaussian distribution. 

\section{Auxiliary Results and Properties\label{sec:Auxiliary-Results}}

\subsection{Structure for Certain Correlation Matrices\label{subsec:MatrixStructures}}

While the elements of $\gamma$ depend on the correlation matrix in
a nonlinear way, there are some interesting correlation structures
that do carry over to the matrix $G=\log C$, and hence $\gamma$.
First, we consider the case with an equicorrelation matrix and a block-equicorrelation
matrix.
\begin{prop}
\label{prop:EquiCorrelation} Suppose $C$ is an equicorrelation matrix
with correlation parameter $\rho$. Then, all the off-diagonal elements
of matrix $G=\log C$ are identical and equal to \setlength{\belowdisplayskip}{5pt} \setlength{\belowdisplayshortskip}{5pt} \setlength{\abovedisplayskip}{0pt} \setlength{\abovedisplayshortskip}{-10pt}
\begin{equation}
\gamma_{c}=-\tfrac{1}{n}\log\left(\tfrac{1-\rho}{1+\rho(n-1)}\right)=\tfrac{1}{n}\log(1+n\tfrac{\rho}{1-\rho})\in\mathbb{R},\label{eq:corrconst}
\end{equation}
so that $\gamma=\gamma_{c}\iota$, where $\iota\in\mathbb{R}^{n(n-1)/2}$
is the vector of ones, $\iota=(1,\ldots,1)^{\prime}$.
\end{prop}
This result, in conjunction with Theorem \ref{thm:OneToOne}, establishes
that $\gamma_{c}$ is a one-to-one correspondence from the set of
non-singular equicorrelation matrices to the real line, $\mathbb{R}$,
and the inverse mapping is given in closed-form by $\rho(\gamma_{c},n)=\frac{1-e^{-n\gamma_{c}}}{1+(n-1)e^{-n\gamma_{c}}}$.
It follows that $\rho(\gamma_{c},n)$ is confined to the interval
$\bigl(-\frac{1}{n-1},1\bigl)$.

It is easy to verify that if $C$ is a block diagonal matrix, with
equicorrelation diagonal blocks and zero correlation across blocks,
then $\log C$ will have the same block structure, and (\ref{eq:corrconst})
can be used to compute the elements in $\gamma$. In the more general
case where $C$ is a block correlation matrix, then it can be shown
that the logarithmic transformation preserves the block structure.
This is used in \citet{ArchakovHansenLundeMRG} in a multivariate
GARCH model. So that $\log C$ has the same block structure as $C$,
and this transformation provides a simple way to model block correlation
matrices. We illustrate this with the following example\begin{small}\renewcommand*{\arraystretch}{.65}\setlength{\belowdisplayskip}{5pt}  \setlength{\abovedisplayskip}{5pt} 
\[
C=\left(\begin{array}{cccccc}
1.0 & {\color{teal}0.4} & {\color{teal}0.4} & {\color{purple}0.2} & {\color{purple}0.2} & {\color{purple}0.2}\\
{\color{teal}0.4} & 1.0 & {\color{teal}0.4} & {\color{purple}0.2} & {\color{purple}0.2} & {\color{purple}0.2}\\
{\color{teal}0.4} & {\color{teal}0.4} & 1.0 & {\color{purple}0.2} & {\color{purple}0.2} & {\color{purple}0.2}\\
{\color{purple}0.2} & {\color{purple}0.2} & {\color{purple}0.2} & 1.0 & {\color{blue}0.6} & {\color{blue}0.6}\\
{\color{purple}0.2} & {\color{purple}0.2} & {\color{purple}0.2} & {\color{blue}0.6} & 1.0 & {\color{blue}0.6}\\
{\color{purple}0.2} & {\color{purple}0.2} & {\color{purple}0.2} & {\color{blue}0.6} & {\color{blue}0.6} & 1.0
\end{array}\right)\Leftrightarrow\log C=\left(\begin{array}{cccccc}
-.16 & {\color{teal}.349} & {\color{teal}.349} & {\color{purple}.104} & {\color{purple}.104} & {\color{purple}.104}\\
{\color{teal}.349} & -.16 & {\color{teal}.349} & {\color{purple}.104} & {\color{purple}.104} & {\color{purple}.104}\\
{\color{teal}.349} & {\color{teal}.349} & -.16 & {\color{purple}.104} & {\color{purple}.104} & {\color{purple}.104}\\
{\color{purple}.104} & {\color{purple}.104} & {\color{purple}.104} & -.36 & {\color{blue}.553} & {\color{blue}.553}\\
{\color{purple}.104} & {\color{purple}.104} & {\color{purple}.104} & {\color{blue}.553} & -.36 & {\color{blue}.553}\\
{\color{purple}.104} & {\color{purple}.104} & {\color{purple}.104} & {\color{blue}.553} & {\color{blue}.553} & -.36
\end{array}\right).
\]
\end{small}Another interesting class of correlation matrices are
the Toeplitz-correlation matrices, which arise in some models, such
as stationary time series models. For this case, $\log C$ is a bisymmetric
matrix.

\subsection{The Inverse  and other Powers of the Correlation Matrix}

Since $C^{\alpha}=e^{\alpha G}$, it is possible to obtain powers
of $C$ from $\gamma$. For instance, the inverse covariance matrix
is given by $\Sigma^{-1}=\Lambda^{-1}e^{-G}\Lambda^{-1}$, where $\Lambda=\mathrm{diag}(\sigma_{1},\ldots,\sigma_{n})$.
The inverse is, for instance, of interest for computing the partial
correlation coefficients and in portfolio choice problems. Some estimation
methods impose sparsity on $\Sigma^{-1}$. While it is not simple
to impose sparsity on $\Sigma^{-1}$ through $\gamma$, the new parametrization
facilitate new ways to impose a parsimonious structure on $\Sigma$
or $\Sigma^{-1}$, by imposing sparsity (or some other structure)
on $\gamma$ directly. 

\subsection{The Jacobian $\partial\varrho/\partial\gamma$}

Next we establish a result that shows that $\partial\varrho/\partial\gamma=\partial\mathrm{vecl}[C]/\partial\mathrm{vecl}[G]$
has a relatively simple expression. This is convenient for inference,
such as computation of standard errors, and for the construction of
dynamic GARCH-type models, such as a score-driven model for $\gamma=\mathrm{vecl}G$,
see \citet{CrealKoopmanLucasGAS_JAE}, and for the construction of
parameter stability tests, such as that of \citet{Nyblom89}. 
\begin{prop}
\label{prop:dC/dG}We have $\tfrac{\partial\varrho}{\partial\gamma}=E_{l}\Bigl(I-AE_{d}^{\prime}\Bigl(E_{d}AE_{d}^{\prime}\Bigl){}^{-1}E_{d}\Bigl)A(E_{l}+E_{u})^{\prime}$,
where $A=\partial\mathrm{vec}C/\partial\mathrm{vec}G$ and the matrices
$E_{l}$, $E_{u}$ and $E_{d}$ are elimination matrices, such that
$\mathrm{vecl}M=E_{l}\mathrm{vec}M$, $\mathrm{vecl}M^{\prime}=E_{u}\mathrm{vec}M$
and $\mathrm{diag}M=E_{d}\mathrm{vec}M$ for any square matrix $M$
of the same size as $C$. 
\end{prop}
The matrix, $A$, is the same matrix that appeared in the asymptotic
distribution for $\hat{\gamma}$, see (\ref{eq:avar-gamma}). In the
Web Appendix we compute $\partial\varrho/\partial\gamma$ for two
correlation matrices: A $10\times10$ Toeplitz correlation matrix
and one based on the empirical correlation matrix for the 10 industry
portfolios in the Kenneth R. French data library. The two have a very
similar structure.

\subsection{Results for Covariance Matrices with Known Diagonal Elements\label{subsec:Results-for-Covariance-known-diag}}

Some of our results for correlation matrices, apply equally to covariance
matrices with known diagonal elements, and these could be useful in
some applications that involve the matrix logarithm of covariance
matrices. In Corollary \ref{cor:CovarianceMatrix} we state the extensions
to this situation.
\begin{cor}
\label{cor:CovarianceMatrix}For any real symmetric matrix, $A\in\mathbb{R}^{n\times n}$,
and any vector, $v\in\mathbb{R}^{n}$ with strictly positive elements,
there exists a unique vector, $x^{\ast}\in\mathbb{R}^{n}$, such that
$\Sigma=e^{A[x^{\ast}]}$ is a covariance matrix with diagonal $\mathrm{diag}(\Sigma)=v$.
Moreover, $x^{\ast}=\lim_{k\rightarrow\infty}x_{(k)}$, where $x_{(k+1)}=x_{(k)}+[\log v-\log\mathrm{diag}(e^{A[x_{(k)}]})]$,
for $k=0,1,2,\ldots$, with an arbitrary initial vector $x_{(0)}\in\mathbb{R}^{n}$. 
\end{cor}

\section{Properties of the Algorithm for the Inverse Mapping, $C(\gamma)$\label{sec:Algorithm}}

The algorithm that reconstructs the correlation matrix, $C$, from
$\gamma$ converges exponentially fast, and its complexity is of order
$O(n^{3}\log n)$. This follows, as we show below, from the fact that
the number of required iterations is of order $\log n$, and because
each iteration entails a matrix exponential evaluation which is of
order $O(n^{3})$, see e.g. \citet{Lu:1998}.

Let $K_{\delta}=\inf\{k:||x_{(k+1)}-x_{(k)}||_{p}\leq\delta\}$ be
the number of iterations required for convergence for some $p$-norm
and some threshold $\delta>0$. From the contraction property  it
follows that $||x_{(k+1)}-x_{(k)}||_{p}\leq L||x_{(k)}-x_{(k-1)}||_{p}\leq L^{k}||x_{(1)}-x_{(0)}||_{p}$,
for $k=1,2,\ldots$, where $L\in[0,1)$ is the Lipschitz constant
given from the contraction. So the number of iterations $k$ can be
bounded from above by $k\leq c_{L}(\log||x_{(1)}-x_{(0)}||_{p}-\log||x_{(k)}-x_{(k-1)}||_{p})$,
where $c_{L}=-\tfrac{1}{\log L}>0$ depends on the Lipschitz constant.
Since $||x||_{p}\leq(n\cdot\max_{1\leq i\leq n}|x^{(i)}|^{p})^{1/p}=n^{1/p}||x||_{\infty}$,
we have

\begin{equation}
K_{\delta}\leq c_{L}(\tfrac{\log n}{p}+\log||x_{(1)}-x_{(0)}||_{\infty}-\log\delta)=O(\log n).\label{eq:IterationBound}
\end{equation}
Note that the number of required iterations may be more sensitive
to the structure of $C$ (through the Lipschitz constant) than the
dimension of $C$. The Lipschitz constant approaches one as $C$ approaches
singularity.  The number of iterations is less sensitive to the choice
of initial vector $x_{(0)}$, but it is useful to know that the elements
of $x^{\ast}$ are non-positive.
\begin{lem}
\label{lem:DiagonalA}The diagonal elements of $\log C$ are non-positive
for any $C\in\mathcal{C}_{n}$.
\end{lem}
The result in (\ref{eq:IterationBound}) is illustrated in Figure
\ref{fig:conv1} where we recover the correlation matrix from $\gamma$
using the algorithm in Corollary \ref{cor:Algorithm}. The true $C$
has a Toeplitz structure, $C_{ij}=\rho^{|i-j|}$, $i,j=1,\ldots,n$,
for $n=3,\ldots,100$ and $\rho\in\{0.5,0.9,0.99\}$. The number of
iterations needed for $||x_{(k)}-x_{(k-1)}||_{2}<\delta=10^{-8}\sqrt{n}$
increases with the dimension at a rate that is consistent with $\log n$.
The number of iterations is sensitive to the correlation structure.
For instance, when $C$ is almost singular ($\rho=0.99$), the number
of iterations is about five times that of a moderately correlated
correlation matrix ($\rho=0.5$). The reason is that a near zero eigenvalue
of $C$ translates into a Lipschitz constant close to one. To illustrate
the sensitivity to the starting value, we use 1,000 different starting
values, $x_{(0)}$, where the elements of $x_{(0)}$ are drawn independently
from the negative half-normal distribution with scale $\sigma=10$
(i.e. $-|Z|$ with $Z\sim N(0,100)$). The shaded bands depict the
dispersion in the number of iterations (average $\pm2$ standard deviations).
The dispersion is relatively modest which verifies that the algorithm
is relatively insensitive to the initial value, $x^{(0)}$.
\begin{figure}[H]
\begin{centering}
\includegraphics[scale=0.7]{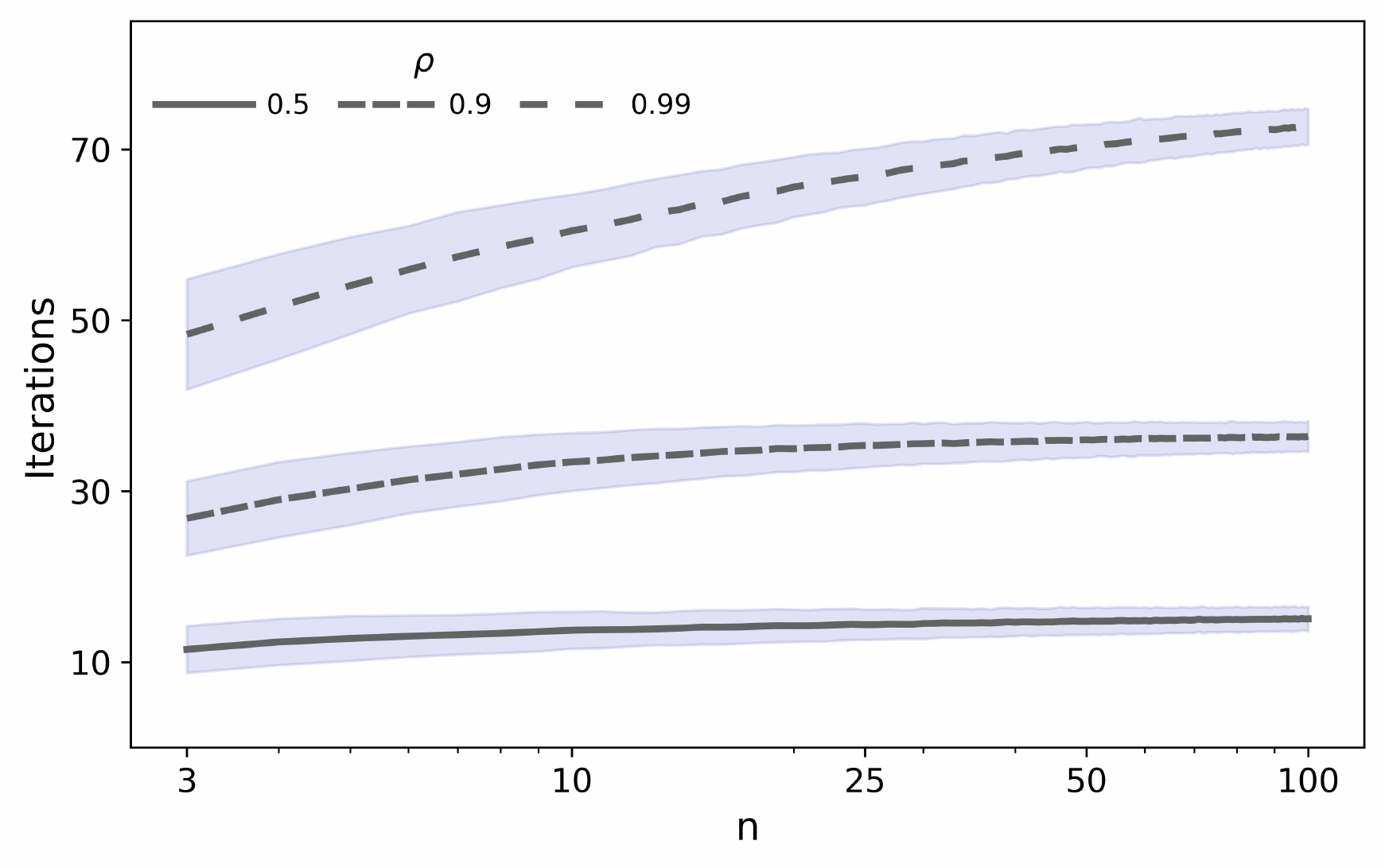}
\par\end{centering}
\caption{{\small{}Number of iterations needed for convergence at threshold
$\delta=10^{-8}\sqrt{n}$, when }$C_{ij}=\rho^{|i-j|}${\small{} }$i,j=1,\ldots,n$,
for $n=3,\ldots,100${\small{}, using random initial value, $x_{(0)}$.
Black lines correspond to the average number of iterations required
for convergence, and the shaded bands ($\pm$2 standard deviations)
show the variation resulting from the different starting values.}\label{fig:conv1}}
\end{figure}

The results in Figure \ref{fig:conv1} are not specific to the Toeplitz
structure for $C$. In a second design, we generate 50,000 distinct
correlation matrices for each of the dimensions, $n\in\{5,10,25\}$.
This is done by generating random vectors, $\gamma$, where each element
in $\gamma$ is uniformly distributed on the interval $[-b_{n},b_{n}]$.
The constant, $b_{n}$, is chosen to provide a sufficiently wide range
of the smallest eigenvalue of $C$, denoted $\lambda_{\min}$, and
the spectral radius of $\nabla g(x^{\ast})$, denoted $\nu_{\max}$.
The Lipschitz constant for the contraction, $g(x)$, is approximately
equal to $\nu_{\max}$, so we should expect $-1/\log\nu_{\max}\simeq c_{L}$
to be linearly related to (the bound on) the number of iterations. 

The number of iterations needed for convergence is shown in Figure
\ref{fig:conv_spectrum}, for $n=5$, $n=10$, and $n=25$, using
scatter plots against three characteristics of $C$. The starting
value is $x_{(0)}=0\in\mathbb{R}^{n}$ in all simulations and $\delta=10^{-8}\sqrt{n}$
was used as the tolerance level. 

The left panels reveal a fairly tight linear relationship between
the number of iterations and $-1/\log\nu_{\max}$ ($\approx c_{L}$).
Similarly, $\lambda_{\max}$ and $\gamma_{\max}$, which are easier
to compute, are also related to the number of iterations, albeit not
as tightly as $\nu_{\max}$.
\begin{figure}[H]
\begin{centering}
\includegraphics[scale=0.7]{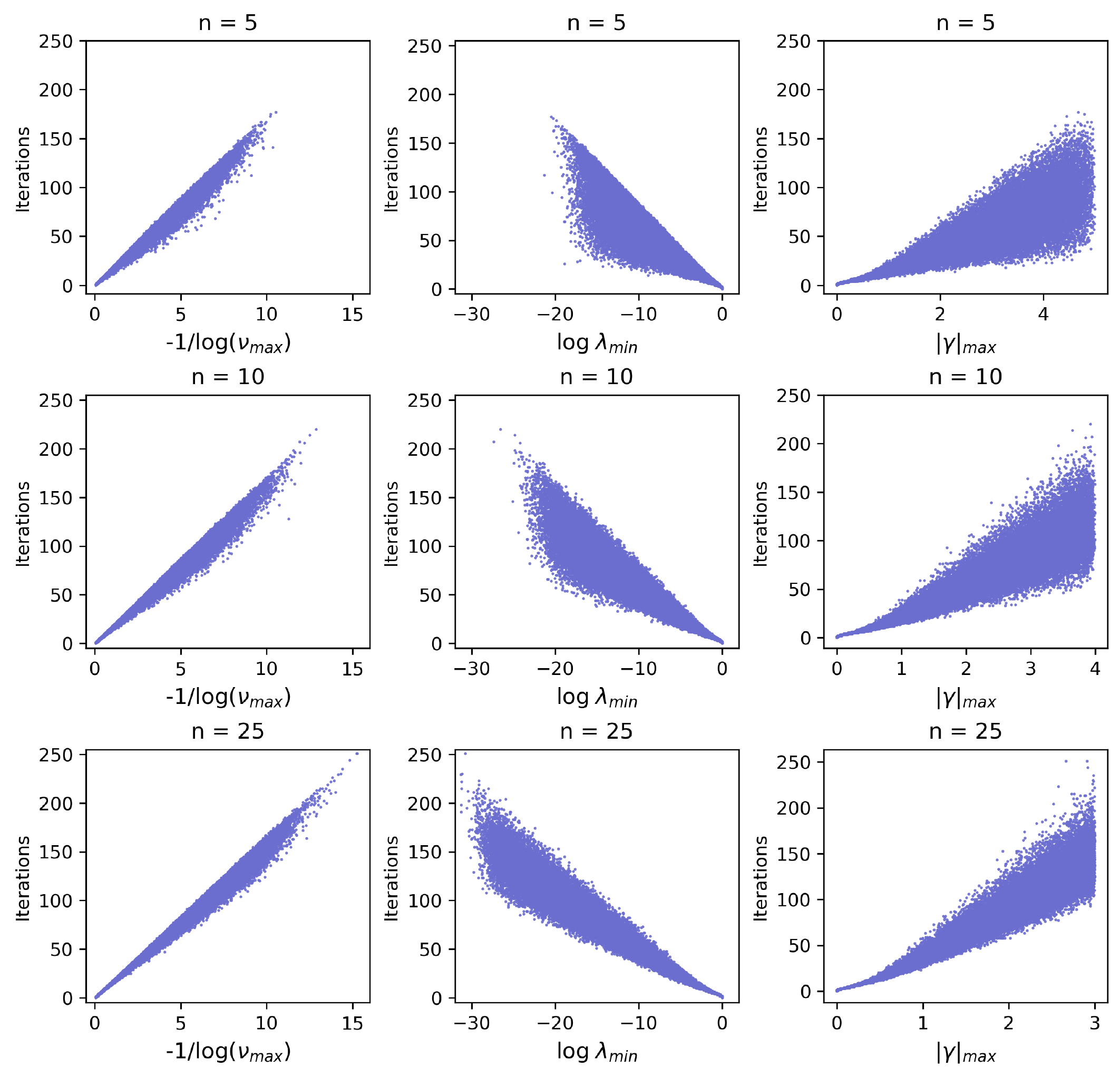}
\par\end{centering}
\caption{The number of iterations needed for convergence plotted against three
characteristics of $C$. The left panels plots the number of iterations
against $-1/\log\nu_{\max}\simeq c_{L}$. The smallest eigenvalue
of $C$ (middle panels) and the largest $|\gamma_{i}|$ (right panels)
are also useful indicators.\label{fig:conv_spectrum}}
\end{figure}

\section{Concluding Remarks\label{sec:Concluding-Remarks}}

In this paper, we have shown that the space of non-singular $n\times n$
correlation matrices is one-to-one with $\mathbb{R}^{n(n-1)/2}$.
A non-singular covariance matrix can therefore be parametrized by
the $n$ (log-)variances and the vector, $\gamma(C)$, which has unrestricted
domain in $\mathbb{R}^{n(n-1)/2}$. This opens new ways to model correlation
and covariance matrices where positive definiteness is an intrinsic
property. For instance, in multivariate GARCH models, as explored
in \citet{ArchakovHansenLundeMRG}. The transformation can be used
to specify probability distributions on correlation and covariance
matrices. Any distribution on $\mathbb{R}^{n(n-1)/2}$ induces a distribution
on the space of positive definite correlation matrices, $\mathcal{C}$.
This could be used in multivariate stochastic volatility modeling,
and defines a new approach to specifying Bayesian priors on $\mathcal{C}$.

We have derived results for the asymptotic distribution of $\gamma(\hat{C})$.
Much is known about the finite sample properties when $n=2$, because
$\gamma(C)$ is identical to the Fisher transformation in this case.
The Fisher transformation has variance stabilizing and skewness eliminating
properties. The variance stabilizing property does not carry over
to the case $n>2$. However, simulation results suggest that it continues
to have skewness reducing properties, and that the empirical distribution
of $\gamma(\hat{C})$ (in a classical setting) is well approximated
by a Gaussian distribution even in small samples. Moreover, the elements
of $\gamma(\hat{C})$ tend to be weakly dependent, as suggested by
the asymptotic results in Table \ref{tab:AvarAcorr}. This makes the
transformation potentially useful for regularization, see \citet{Pourahmadi2011},
and inference. These attributes tend to deteriorate as $C$ approaches
singularity. This is not unexpected, because it is also true for the
Fisher transformation when the correlation is close to $\pm1$.

The inverse mapping, $C(\gamma)$ is not given in closed-form when
$n>2$, except in some special cases. Instead, we proposed a fast
algorithm to evaluate $C(\gamma)$, and showed that its numerical
complexity is of order $O(n^{3}\log n)$, where $n\times n$ is the
dimension of $C$.%

{\footnotesize{}\bibliographystyle{agsm}
\bibliography{prh}
}{\footnotesize\par}

\appendix

\section*{Appendix of Proofs}

\setcounter{equation}{0}\renewcommand{\theequation}{A.\arabic{equation}}
\setcounter{thm}{0}\renewcommand{\thethm}{A.\arabic{thm}}
\setcounter{lem}{0}\renewcommand{\thelem}{A.\arabic{lem}}
\setlength{\belowdisplayskip}{5pt} \setlength{\belowdisplayshortskip}{5pt} \setlength{\abovedisplayskip}{5pt} \setlength{\abovedisplayshortskip}{5pt}

We prove $g$ is a contraction by deriving its Jacobian, $J(x)$,
and showing that all its eigenvalues are less than one in absolute
value. Since $g(x)=x-\log\delta(x),$ where $\delta(x)=\text{diag}(e^{G[x]})$,
an intermediate step towards the Jacobian for $g$, is to derive the
Jacobian for $\delta(x)$. To simplify notation, we sometimes suppress
the dependence on $x$ for some terms. For instance, we sometimes
write $\delta_{i}$ to denote the $i$-th element of the vector $\delta(x)$.
It follows that $[J(x)]_{i,j}=\frac{\partial[g(x)]_{i}}{\partial x_{j}}=1_{\{i=j\}}-\frac{1}{\delta_{i}}\frac{\partial[\delta(x)]_{i}}{\partial x_{j}},$
so that $J(x)=I-[D(x)]^{-1}H(x)$, where $D(x)=\mathrm{diag}(\delta_{1},\ldots,\delta_{n})$
is a diagonal matrix and $H(x)$ is the Jacobian matrix of $\delta(x)$,
derived below.

Let $G[x]=Q\Lambda Q^{\prime}$, where $\Lambda$ is the diagonal
matrix with the eigenvalues, $\lambda_{1},\ldots,\lambda_{n}$, of
$G[x]$ and $Q$ is an orthonormal matrix (i.e. $Q^{\prime}=Q^{-1}$)
with the corresponding eigenvectors. From \citet{LintonMcCrorie:1995},
we have $\mathrm{d}\text{vec}\,e^{G[x]}=A(x)\,\mathrm{d}\text{vec}G[x]$,
where
\begin{equation}
A(x)=(Q\otimes Q)\Xi\bigl(Q\otimes Q\bigl)^{\prime},\label{eq:JacobianA-expression}
\end{equation}
is and $n^{2}\times n^{2}$ matrix and $\Xi$ is the $n^{2}\times n^{2}$
diagonal matrix with elements given by
\begin{equation}
\xi_{ij}=\Xi_{(i-1)n+j,(i-1)n+j}=\begin{cases}
e^{\lambda_{i}}, & \text{if}\qquad\lambda_{i}=\lambda_{j}\\
\tfrac{e^{\lambda_{i}}-e^{\lambda_{j}}}{\lambda_{i}-\lambda_{j}}, & \text{if}\qquad\lambda_{i}\neq\lambda_{j}
\end{cases}\label{eq:xi_defined}
\end{equation}
for $i=1,\ldots,n$ and $j=1,\ldots,n$. Evidently, we have $\xi_{ij}=\xi_{ji}$,
for all $i$ and $j$. Moreover, $A(x)$ is a symmetric positive definite
matrix, because all the diagonal elements of $\Xi$ are strictly positive. 

Our interest concerns $\delta(x)=\text{diag}[e^{G[x]}]$ (a subset
of the elements of $\text{vec}[e^{G[x]}]$) so the Jacobian of $\delta(x)$,
denoted $H(x)$, is a principal sub-matrix of $A(x)$, defined by
the elements $[A(x)]_{l,m}$, $l,m=(i-1)n+i$, for $i=1,\ldots,n$.
Thus
\begin{align}
[H(x)]_{i,j} & =\frac{\partial\delta(x)_{i}}{\partial x_{j}}=(e_{i}\otimes e_{i})^{\prime}\bigl(Q\otimes Q\bigl)\Xi\bigl(Q\otimes Q)^{\prime}(e_{j}\otimes e_{j})\nonumber \\
 & =\bigl(e_{i}^{\prime}Q\otimes e_{i}^{\prime}Q\bigl)\Xi\bigl(Q^{\prime}e_{j}\otimes Q^{\prime}e_{j}\bigl)=\bigl(Q_{i,.}\otimes Q_{i,.}\bigl)\Xi\bigl(Q_{j,.}\otimes Q_{j,.}\bigl)^{\prime}\label{eq:H_ij}\\
 & =\sum_{k=1}^{n}\sum_{l=1}^{n}q_{ik}q_{jk}q_{il}q_{jl}\xi_{kl},\nonumber 
\end{align}
where $e_{i}$ is a $n\times1$ unit vector with one at the $i$-th
position and zeroes otherwise and $Q_{i,.}$ denotes the $i$-th row
of $Q$.

Interestingly, the Jacobian of $g$ is such that $J(x)\iota=0$, so
that the vector of ones, $\iota$, is an eigenvector of $J(x)$ associated
with the eigenvalue $0$, i.e. $J(x)$ has reduced rank. Because the
$i$-th row of $J(x)$ times $\iota$ reads
\[
1-\sum_{j=1}^{n}\frac{1}{\delta_{i}}\sum_{k=1}^{n}\sum_{l=1}^{n}q_{ik}q_{jk}q_{il}q_{jl}\xi_{kl}=1-\frac{1}{\delta_{i}}\sum_{k=1}^{n}\sum_{l=1}^{n}q_{ik}q_{il}\xi_{kl}\sum_{j=1}^{n}q_{jk}q_{jl}=1-\frac{1}{\delta_{i}}\sum_{k=1}^{n}q_{ik}^{2}\xi_{kk}=0,
\]
due to $\sum_{k=1}^{n}q_{ik}q_{jk}=1_{\{i=j\}}$. 

\subsection*{Proof that $g$ is a Contraction: Lemma \ref{lem:gContraction}}

We now want to prove that the mapping $g(x)$ is a contraction. In
order to show this, it is sufficient to demonstrate that all eigenvalues
of the corresponding Jacobian matrix $J(x)$ are below one in absolute
values for any real vector $x$. First we establish a number of intermediate
results.
\begin{lem}
\label{lem:TwoFunctions}$(i)$ $e^{y}-y-1>0$ for all $y\neq0$,
and $(ii)$ $1+e^{y}-\frac{2}{y}(e^{y}-1)>0$ for $y\neq0$.
\end{lem}
\begin{proof}
The first and second derivatives of $f(y)=e^{y}-y-1$ show that $f$
is strictly convex with unique minimum, $f(0)=0$, which proves $(i)$.
Next we prove $(ii)$. Now let $f(y)=1+e^{y}-\frac{2}{y}(e^{y}-1)$.
Its first derivative is given by $f'(y)=e^{y}y^{-2}g(y)$, where $g(y)=y^{2}-2y+2-2e^{-y}$,
so that $f^{\prime}(y)<0$ for $y<0$ and $f^{\prime}(y)>0$ for $y>0$.
Since $\lim_{y\rightarrow0}f(y)=0$ (by l'Hospital's rule) the result
follows. 
\end{proof}
From the definition, (\ref{eq:xi_defined}), it follows that $\xi_{ij}=\xi_{ii}=\xi_{jj}$
whenever $\lambda_{i}=\lambda_{j}$. When $\lambda_{i}\neq\lambda_{j}$
we have the following results for the elements of $\Xi$:
\begin{lem}
\label{lem:xi-inequalities}If $\lambda_{i}<\lambda_{j}$, then $\xi_{ii}<\xi_{ij}<\xi_{jj}$
and $2\xi_{ij}<\xi_{ii}+\xi_{jj}$.
\end{lem}
\begin{proof}
From the definition, (\ref{eq:xi_defined}), $\xi_{ij}-\xi_{ii}=\tfrac{e^{\lambda_{j}}-e^{\lambda_{i}}}{\lambda_{j}-\lambda_{i}}-e^{\lambda_{i}}=e^{\lambda_{i}}(\tfrac{e^{\lambda_{j}-\lambda_{i}}-1}{\lambda_{j}-\lambda_{i}}-1)=e^{\lambda_{i}}\tfrac{e^{\lambda_{j}-\lambda_{i}}-1-(\lambda_{j}-\lambda_{i})}{\lambda_{j}-\lambda_{i}}>0$,
where the numerator is positive by Lemma \ref{lem:TwoFunctions} ($i$).
So are $e^{\lambda_{i}}$ and $\lambda_{j}-\lambda_{i}$, which proves
$\xi_{ij}>\xi_{ii}$. Analogously, $\xi_{jj}-\xi_{ij}=e^{\lambda_{j}}-\tfrac{e^{\lambda_{j}}-e^{\lambda_{i}}}{\lambda_{j}-\lambda_{i}}=e^{\lambda_{j}}(1-\tfrac{1-e^{\lambda_{i}-\lambda_{j}}}{\lambda_{j}-\lambda_{i}})=e^{\lambda_{j}}\tfrac{-(\lambda_{i}-\lambda_{j})-1+e^{\lambda_{i}-\lambda_{j}}}{\lambda_{j}-\lambda_{i}}>0$,
because all terms are positive, where we again used Lemma \ref{lem:TwoFunctions}
($i$). Next, $\xi_{ii}+\xi_{jj}-2\xi_{ij}=e^{\lambda_{i}}+e^{\lambda_{j}}-2\tfrac{e^{\lambda_{i}}-e^{\lambda_{j}}}{\lambda_{i}-\lambda_{j}}=e^{\lambda_{i}}(1+e^{\lambda_{j}-\lambda_{i}}-2\tfrac{e^{\lambda_{j}-\lambda_{i}}-1}{\lambda_{j}-\lambda_{i}})>0$,
where the inequality follows by Lemma \ref{lem:TwoFunctions}.$ii$,
because $\lambda_{i}\neq\lambda_{j}$.
\end{proof}
\begin{lem}
\label{lem:J-tilde-expression}$J(x)$ and $\tilde{J}(x)=I-D^{-\frac{1}{2}}HD^{-\frac{1}{2}}$
have the same eigenvalues, where
\[
\tilde{J}(x)=\sum_{k=1}^{n-1}\sum_{l=k}^{n}\varphi_{kl}\Bigl(D^{-\frac{1}{2}}u_{kl}u_{kl}^{\prime}D^{-\frac{1}{2}}\Bigl),
\]
with $u_{kl}=Q_{\cdot,k}\odot Q_{\cdot,l}\in\mathbb{R}^{n}$ and $\varphi_{kl}=\xi_{kk}+\xi_{ll}-2\xi_{kl}$.
\end{lem}
\begin{proof}
For a vector $y$ and a scalar $\nu$, $Jy=\nu y\Leftrightarrow\tilde{J}w=\nu w$,
where $y=D^{-\frac{1}{2}}w$, because $J=I-D^{-1}H=D^{-\frac{1}{2}}(I-D^{-\frac{1}{2}}HD^{-\frac{1}{2}})D^{\frac{1}{2}}=D^{-\frac{1}{2}}\tilde{J}D^{\frac{1}{2}}$.
Next, we turn to the expression for $\tilde{J}$. First, note that
$\sum_{k=1}^{n}q_{ik}^{2}\xi_{kk}=\sum_{k=1}^{n}q_{ik}^{2}e^{\lambda_{k}}=Q_{i,\cdot}e^{\Lambda}Q_{i,\cdot}^{\prime}=[e^{Q\Lambda Q^{\prime}}]_{ii}=[e^{G}]_{ii}=\delta_{i}$.
So diagonal elements of $\tilde{J}$ are given by 
\begin{align*}
\tilde{J}_{ii}=1-\frac{H_{ii}}{\delta_{i}} & =\frac{1}{\delta_{i}}\Bigl(\sum_{k=1}^{n}q_{ik}^{2}\xi_{kk}-\sum_{k=1}^{n}\sum_{l=1}^{n}q_{ik}^{2}q_{il}^{2}\xi_{kl}\Bigl)=\frac{1}{\delta_{i}}\Bigl(\sum_{k=1}^{n}q_{ik}^{2}\xi_{kk}-\sum_{k=1}^{n}q_{ik}^{2}q_{ik}^{2}\xi_{kk}-2\sum_{k=1}^{n-1}\sum_{l=k}^{n}q_{ik}^{2}q_{il}^{2}\xi_{kl}\Bigl)\\
 & =\frac{1}{\delta_{i}}\Bigl(\sum_{k=1}^{n}q_{ik}^{2}\xi_{kk}(1-q_{ik}^{2})-2\sum_{k=1}^{n-1}\sum_{l=k}^{n}q_{ik}^{2}q_{il}^{2}\xi_{kl}\Bigl)=\frac{1}{\delta_{i}}\Bigl(\sum_{k=1}^{n}q_{ik}^{2}\xi_{kk}\sum_{\substack{l=1\\
l\neq k
}
}^{n}q_{il}^{2}-2\sum_{k=1}^{n-1}\sum_{l=k}^{n}q_{ik}^{2}q_{il}^{2}\xi_{kl}\Bigl)\\
 & =\frac{1}{\delta_{i}}\Bigl(\sum_{k=1}^{n-1}\sum_{l=k}^{n}q_{ik}^{2}q_{il}^{2}(\xi_{kk}+\xi_{ll})-2\sum_{k=1}^{n-1}\sum_{l=k}^{n}q_{ik}^{2}q_{il}^{2}\xi_{kl}\Bigl)=\frac{1}{\delta_{i}}\sum_{k=1}^{n-1}\sum_{l=k}^{n}q_{ik}^{2}q_{il}^{2}\varphi_{kl},
\end{align*}
where we used (\ref{eq:H_ij}). Similarly for the off-diagonal elements
we have
\begin{align*}
\tilde{J}_{ij}=-\frac{H_{ij}}{\sqrt{\delta_{i}\delta_{j}}} & =-\frac{1}{\sqrt{\delta_{i}\delta_{j}}}\sum_{k=1}^{n}\sum_{l=1}^{n}q_{ik}q_{jk}q_{il}q_{jl}\xi_{kl}=-\frac{1}{\sqrt{\delta_{i}\delta_{j}}}\Bigl(\sum_{k=1}^{n}q_{ik}^{2}q_{jk}^{2}\xi_{kk}+2\sum_{k=1}^{n-1}\sum_{l=k}^{n}q_{ik}q_{jk}q_{il}q_{jl}\xi_{kl}\Bigl)\\
 & =-\frac{1}{\sqrt{\delta_{i}\delta_{j}}}\Bigl(\sum_{k=1}^{n}q_{ik}q_{jk}\Big(-\sum_{\substack{l=1\\
l\neq k
}
}^{n}q_{il}q_{jl}\Big)\xi_{kk}+2\sum_{k=1}^{n-1}\sum_{l=k}^{n}q_{ik}q_{jk}q_{il}q_{jl}\xi_{kl}\Bigl)\\
 & =-\frac{1}{\sqrt{\delta_{i}\delta_{j}}}\Bigl(-\sum_{k=1}^{n-1}\sum_{l=k}^{n}q_{ik}q_{jk}q_{il}q_{jl}(\xi_{kk}+\xi_{ll})+2\sum_{k=1}^{n-1}\sum_{l=k}^{n}q_{ik}q_{jk}q_{il}q_{jl}\xi_{kl}\Bigl)\\
 & =\frac{1}{\sqrt{\delta_{i}\delta_{j}}}\sum_{k=1}^{n-1}\sum_{l=k}^{n}q_{ik}q_{jk}q_{il}q_{jl}\varphi_{kl}.
\end{align*}
In the derivations above we used that $\sum_{k=1}^{n}q_{ik}q_{jk}=1_{\{i=j\}}$,
since $Q^{\prime}Q=QQ^{\prime}=I$.
\end{proof}
\noindent\textbf{Proof of Lemma \ref{lem:gContraction}.} Because
$A(x)$ is symmetric and positive definite, then so is the principal
sub-matrix, $H(x)$. Consequently, $M=D^{-\frac{1}{2}}H(x)D^{-\frac{1}{2}}$
is symmetric and positive definite. Thus, any eigenvalue, $\mu$ of
$M$ is strictly positive. So if $\nu$ is an eigenvalue of $\tilde{J}(x)=I-D^{-\frac{1}{2}}HD^{-\frac{1}{2}}$,
then $\nu=1-\mu$ where $\mu$ is an eigenvalue of $M$, from which
it follows that all eigenvalues of $\tilde{J}$ are strictly less
than 1.

Consider a quadratic form of $\tilde{J}$ with an arbitrary vector
$z\in\mathbb{R}^{n}$. Using Lemma \ref{lem:J-tilde-expression},
it follows that any quadratic form is bounded from below by

\[
z^{\prime}\tilde{J}z=\sum_{k=1}^{n-1}\sum_{l=k}^{n}\varphi_{kl}\Bigl(z^{\prime}D^{-\frac{1}{2}}u_{kl}u_{kl}^{\prime}D^{-\frac{1}{2}}z\Bigl)=\sum_{k=1}^{n-1}\sum_{l=k}^{n}\varphi_{kl}\Bigl(z^{\prime}D^{-\frac{1}{2}}u_{kl}\Bigl)^{2}\geq0,
\]
because $\varphi_{kl}>0$ by Lemma \ref{lem:xi-inequalities}. Hence,
$\tilde{J}$ is positive semi-definite and $\nu_{i}\geq0$, for all
$i=1,\ldots,n$. 

Finally, since $J(x)$ and $\tilde{J}(x)$ have the same eigenvalues,
it follows that all eigenvalues of $J(x)$ lie within the interval
$[0,1)$, which proves that $g(x)$ is a contraction. $\square$

\noindent\textbf{Proof of Theorem \ref{thm:OneToOne}.} The Theorem
is equivalent to the statement that for any symmetric matrix $G$,
there always exists a unique solution to $g(x)=x$. This follows from
Lemma \ref{lem:gContraction} and Banach's fixed point theorem. $\square$

\noindent\textbf{Proof of Proposition \ref{prop:Permutate}.} We
have $Y=PX$, for some permutation matrix, $P$, so that $C_{y}=PC_{x}P^{\prime}$.
Let $C_{x}=Q\Lambda Q^{\prime}$ be the spectral decomposition of
$C_{x}$, such that $\log C_{x}=Q\log\Lambda_{x}Q^{\prime}$, where
$Q^{\prime}Q=I$ and $\Lambda$ is a diagonal matrix. So $C_{y}=PC_{x}P^{\prime}=PQ\Lambda Q^{\prime}P^{\prime}$,
where $Q^{\prime}P^{\prime}PQ=Q^{\prime}Q=I$. The first equality
uses the fact that $P$ is a permutation matrix. Therefore, $C_{y}=(PQ)\Lambda(PQ)^{\prime}$
is the spectral decomposition of $C_{y}$ and $\log C_{y}=(PQ)\log\Lambda(PQ)^{\prime}=P[\log C_{x}]P^{\prime}$. 

Next, let the $i$-th and $j$-th rows of $P$ be the $r$-th and
$s$-th unit vectors, $e_{r}^{\prime}$ and $e_{s}^{\prime}$, respectively.
Then we have $[\log C_{y}]_{ij}=[\log C_{x}]_{rs}$ and, by symmetry,
\[
[\log C_{y}]_{ij}=[\log C_{y}]_{ji}=[\log C_{x}]_{rs}\mathbf{=}[\log C_{x}]_{sr},
\]
which shows that $\gamma_{y}$ is simply a permutation of the elements
in $\gamma_{x}$. $\square$

\noindent\textbf{Proof of Proposition \ref{prop:EquiCorrelation}.}
An equicorrelation matrix can be written as $C=(1-\rho)I_{n}+\rho U_{n}$,
where $I_{n}\in\mathbb{R}^{n\times n}$ is the identity matrix and
$U_{n}\in\mathbb{R}^{n\times n}$ is a matrix of ones. Using the Sherman--Morrison
formula, we can obtain the inverse, $C^{-1}=\frac{1}{1-\rho}(I_{n}-\frac{\rho}{1+(n-1)\rho}U_{n})$,
so that
\begin{equation}
G=\log C=-\log(C^{-1})=-\log(\tfrac{1}{1-\rho}I_{n})-\log(I_{n}-\tfrac{\rho}{1+(n-1)\rho}U_{n}).\label{eq:ec2}
\end{equation}
Because the first term is a diagonal matrix, the off-diagonal elements
of $G$ are determined only by the second term, which equals
\begin{equation}
-\log(I_{n}-\varphi U_{n})=-\sum_{k=1}^{\infty}(-1)^{k+1}\tfrac{(-\varphi U_{n})^{k}}{k}=-\left[\frac{1}{n}\sum_{k=1}^{\infty}(-1)^{k+1}\tfrac{(-n\varphi)^{k}}{k}\right]U_{n}=-\frac{1}{n}\log(1-n\varphi)U_{n},\label{eq:ec3}
\end{equation}
where $\varphi=\rho/(1+(n-1)\rho)$ and we have used the fact that
$U_{n}^{k}=n^{k-1}U_{n}$. It now follows that
\[
G_{ij}=\gamma_{c}=-\frac{1}{n}\log(1-\tfrac{n\rho}{1+(n-1)\rho})=-\frac{1}{n}\log\tfrac{1-\rho}{1+\rho(n-1)}=-\frac{1}{n}\log\tfrac{1-\rho}{1+\rho(n-1)},\quad\text{for all }i\neq j
\]
for all $i$ and $j$, such that $i\neq j$. $\square$

\noindent\textbf{Proof of Proposition \ref{prop:dC/dG}.} From Theorem
\ref{thm:OneToOne} it follows that the diagonal, $x=\mathrm{diag}G$,
is fully characterized by the off-diagonal elements, $y=\mathrm{vecl}G=\mathrm{vecl}G^{\prime}$,
and we may write $x=x(y)$. For the off-diagonal elements of the correlation
matrix, $z=\mathrm{vecl}C\mathrm{=vecl}C^{\prime}$, we have $z=z(x,y)=z(x(y),y)$,
since $C=e^{G}$, and it follows that
\begin{equation}
\frac{dz(x,y)}{dy}=\frac{\partial z(x,y)}{\partial x}\frac{dx(y)}{dy}+\frac{\partial z(x,y)}{\partial y}.\label{eq:d2}
\end{equation}
With $A(x,y)=d\mathrm{vec}C/d\mathrm{vec}G$ and the definitions of
$E_{l}$ and $E_{u}$, the second term is given by:
\begin{equation}
\frac{\partial z(x,y)}{\partial y}=E_{l}A(x,y)E_{l}^{\prime}+E_{l}A(x,y)E_{u}^{\prime}.\label{eq:d3}
\end{equation}
The expression has two terms because a change in an element of $y$
affects two symmetric entries in the matrix $G$. Similarly, for the
first part of the first term in (\ref{eq:d2}) we have,
\begin{equation}
\frac{\partial z(x,y)}{\partial x}=E_{l}A(x,y)E_{d}^{\prime},\label{eq:d4}
\end{equation}
and what remains is to determine $\frac{dx(y)}{dy}$. For this purpose
we introduce $D(x,y)=\mathrm{diag}[e^{G(x,y)}]-\iota$ which implicitly
defines the relation between $x$ and $y$. The requirement that $e^{G}$
is a correlation matrix, is equivalent to $D(x,y)=0$. Next, let $\frac{\partial D}{\partial x}$
and $\frac{\partial D}{\partial y}$ denote the Jacobian matrices
of $D(x,y)$ with respect to $x$ and $y$, respectively. These Jacobian
matrices have dimensions $n\times n$ and $n\times n(n-1)/2$, respectively,
and can also be expressed in terms of matrix $A(x,y)$, as follows
\begin{align*}
\frac{\partial D}{\partial x}=E_{d}A(x,y)E_{d}^{\prime}, & \qquad\frac{\partial D}{\partial y}=E_{d}A(x,y)E_{l}^{\prime}+E_{d}A(x,y)E_{u}^{\prime}.
\end{align*}
Note that $\frac{\partial D}{\partial x}$ is a principal sub-matrix
of positive definite matrix $A$ and, hence, is an invertible matrix.
Therefore, from the Implicit Function Theorem it follows
\begin{equation}
\frac{dx(y)}{dy}=-\Bigl(\frac{\partial D}{\partial x}\Bigl){}^{-1}\frac{\partial D}{\partial y}=-\Bigl(E_{d}A(x,y)E_{d}^{\prime}\Bigl){}^{-1}\Bigl(E_{d}A(x,y)E_{l}^{\prime}+E_{d}A(x,y)E_{u}^{\prime}\Bigl).\label{eq:d5}
\end{equation}
The results now follows by inserting \eqref{eq:d3}, \eqref{eq:d4}
and \eqref{eq:d5} into \eqref{eq:d2}. $\square$

\noindent\textbf{Proof of Lemma \ref{lem:DiagonalA}.} We have $G=Q\log\Lambda Q^{\prime}$,
where $C=Q\Lambda Q^{\prime}$ is the spectral decomposition of the
correlation matrix. Thus a generic element of $G$ can be written
as $G_{ij}=\sum_{k=1}^{n}q_{ik}q_{jk}\log\lambda_{k}$. By Jensen's
inequality it follows that $G_{ii}=\sum_{k=1}^{n}q_{ik}^{2}\log\lambda_{k}\leq\log\left(\sum_{k=1}^{n}q_{ik}^{2}\lambda_{k}\right)$,
where we used that $\sum_{k=1}^{n}q_{ik}q_{jk}=1_{\{i=j\}}$, because
$Q^{\prime}Q=I$. Finally, since $\sum_{k=1}^{n}q_{ik}^{2}\lambda_{k}=C_{ii}=1$,
it follows that $G_{ii}\leq\log1=0$. $\square$
\end{document}